\documentclass[a4paper, amsfonts, amssymb, amsmath, reprint, showkeys, twoside,superscriptaddress]{revtex4-2}
\usepackage[english]{babel}
\usepackage[utf8]{inputenc}
\usepackage[colorinlistoftodos, color=green!40, prependcaption]{todonotes}
\usepackage{qcircuit}
\usepackage{amsthm}
\usepackage{mathtools}
\usepackage{xcolor}
\usepackage{graphicx}
\graphicspath{{figures/}}
\usepackage[left=0.6in,right=0.6in,top=1in,columnsep=15pt]{geometry} 
\usepackage[T1]{fontenc}
\usepackage{centernot}

\newcommand{\bra}[1]{\langle #1 \rvert}
\newcommand{\ket}[1]{\lvert #1 \rangle}

\newcommand{\tr}[1]{\textrm{Tr}\left \lbrace #1 \right \rbrace}

\newcommand{\im}[1]{\mathbf{Im} \left[ #1 \right]}
\newcommand{\re}[1]{\mathbf{Re} \left[ #1 \right]}
\newcommand{\jj}[1]{\mathbf{J} \left[ #1 \right]}

\newcommand{\onot}[1]{\mathcal{O} \left( #1 \right) }

\newcommand{\ft}[1]{\mathcal{F}_{t\rightarrow \omega}\left[#1\right]}
\newcommand{\ift}[1]{\mathcal{F}^{-1}_{\omega \rightarrow t}\left[#1\right]}
\newcommand{\esp}[1]{\mathbb{E}\left[#1\right]}

\usepackage{glossaries}
\newacronym{QPE}{QPE}{Quantum Phase Estimation}
\newacronym{GQPE}{GQPE}{Generalized Quantum Phase Estimation}
\newacronym{PRO}{PRO}{Path Response Operator}
\newacronym{QFT}{QFT}{Quantum Fourier Transform}
\newacronym{LCU}{LCU}{Linear Combination of Unitaries}
\usepackage[pdftex, pdftitle={Article}, pdfauthor={Author}]{hyperref} 
\usepackage[capitalize,nameinlink]{cleveref}
\hypersetup{
    colorlinks=true, 
    linkcolor=black,  
    citecolor=blue, 
}

\newtheorem{theorem}{Theorem}

\makeatletter
\def\Ddots{\mathinner{\mkern1mu\raise\p@
\vbox{\kern7\p@\hbox{.}}\mkern2mu
\raise4\p@\hbox{.}\mkern2mu\raise7\p@\hbox{.}\mkern1mu}}
\makeatother

\begin{document}
\title{Nonlinear Spectroscopy via Generalized Quantum Phase Estimation}

\newcommand{\be}{\mathbb{E}}

\author{Ignacio Loaiza}
\thanks{These authors contributed equally. \\
danial.motlagh@xanadu.ai\\
ignacio.loaiza@xanadu.ai}
\affiliation{Xanadu. Toronto, ON. M5G 2C8. Canada}
\author{Danial Motlagh}
\thanks{These authors contributed equally. \\ 
danial.motlagh@xanadu.ai\\
ignacio.loaiza@xanadu.ai}
\affiliation{Xanadu. Toronto, ON. M5G 2C8. Canada}
\author{Kasra Hejazi}
\affiliation{Xanadu. Toronto, ON. M5G 2C8. Canada}
\author{Modjtaba Shokrian Zini}
\affiliation{Xanadu. Toronto, ON. M5G 2C8. Canada}
\author{Alain Delgado}
\affiliation{Xanadu. Toronto, ON. M5G 2C8. Canada}
\author{Juan Miguel Arrazola}
\affiliation{Xanadu. Toronto, ON. M5G 2C8. Canada}

\date{\today}

\begin{abstract}
Response theory has a successful history of connecting experimental observations with theoretical predictions. Of particular interest is the optical response of matter, from which spectroscopy experiments can be modelled. However, the calculation of response properties for quantum systems is often prohibitively expensive, especially for nonlinear spectroscopy, as it requires access to either the time-evolution of the system or to excited states. In this work, we introduce a generalized quantum phase estimation framework designed for multi-variate phase estimation. This allows the treatment of general correlation functions enabling the recovery of response properties of arbitrary orders. The generalized quantum phase estimation circuit has an intuitive construction that is linked with a physical process of interest, and can directly sample frequencies from the distribution that would be obtained experimentally. In addition, we provide a single-ancilla modification of the new framework for early fault-tolerant quantum computers. Overall, our framework enables the efficient simulation of spectroscopy experiments beyond the linear regime, such as Raman spectroscopy, having that the circuit cost grows linearly with respect to the order of the target nonlinear response. This opens up an exciting new field of applications for quantum computers with potential technological impact.
\end{abstract}

\maketitle

\onecolumngrid

\glsresetall

\section{Introduction}
Response theory provides a systematic framework for calculations of field-induced properties of molecules and materials through changes in expectation values \cite{intro_to_response}. The most straightforward response functions to compute correspond to the linear response regime. Often expressed through Kubo formulas \cite{kubo_0,kubo_1,kubo_2}, linear response functions are associated with many properties of relevance in nuclear, chemical, and condensed-matter physics. Prime examples are linear absorption and emission spectroscopy \cite{mukamel,boyd}, first-order magnetic susceptibilities of materials \cite{magnetic,magnetic_1,magnetic_2}, thermal responses \cite{thermal_resp}, Hall conductivities \cite{hall}, and dielectric responses \cite{dielectric_protein}. \\

Despite being less explored due to increased computational cost, nonlinear response has a wide range of exciting applications. Nonlinear magnetic responses can be used as a sensitive probe to the structural and electronic properties of materials \cite{magnetic,magnetic_1,magnetic_2}, with direct applications to superconductivity \cite{magnetic_superconductivity} and topological materials \cite{magnetic_topology}. Nonlinear thermal responses can be linked to interesting topological properties of matter \cite{nonlinear_thermal} and thermal transport properties \cite{thermal_transport_1,thermal_transport_2,thermal_transport_3}. The recently discovered nonlinear Hall effect can be used to probe phase transitions in materials, modelling exotic electrical behaviours which appear when two alternating currents are applied \cite{nonlinear_hall_1,nonlinear_hall_2}. With a growing range of applications, the field of nonlinear response is an exciting area for technological innovation. Of particular interest are nonlinear optical responses, which give rise to a wide array of Raman spectroscopies and multiple harmonic generation \cite{mukamel,boyd,2D_spec,microsoft_n_spec}. Nonlinear spectroscopy has a rich history of experimental applications, and will be the main focus in this work. However, we note that the formalism presented in this work can be applied to any type of response. \\

While there is a large body of work dedicated to calculating response quantities on classical computers \cite{nielsen_and_chuang, shankar, susskind, jensen, mukamel, boyd}, such methods are often limited to small system sizes or require the use approximate methods, particularly for nonlinear responses which are harder to compute. This difficulty is due to the fact that unlike the ground state energy problem, response quantities require access either to excited states, which are often more challenging to prepare than the ground state, or to the time-evolution operator. While no general sub-exponential classical algorithms are known for Hamiltonian simulation, simulating time-evolution under a local Hamiltonian can be efficiently implemented on quantum computers \cite{Kassal_2008,qSimulation,simulation1,simulation2}. This offers a promising avenue to overcome the limitations of classical methods for theoretical spectroscopy.\\ 

Overall, little attention has been given to this topic within quantum computing when compared to ground-state calculations. Several works \cite{lr_on_qc,maskara2023programmable, 2D_spec,analog_vib,vibronic,reinholdt2024subspace,microsoft_n_spec,molecular_response,correlation2,lr_with_broadening,response_for_lr} exist on the topic of calculating response properties on quantum computers, with the majority being concerned with linear optical response. Many of the existing approaches calculate correlation functions in the time-domain and obtain the frequency-dependent properties via classical post-processing, which requires sampling of many different times. Recently a method was proposed in \cite{microsoft_n_spec} that uses quantum signal processing \cite{low2017optimal, motlagh2023generalized} for obtaining the frequency-dependent response of arbitrary order spectroscopies . However, their framework estimates one out of exponentially many frequencies at a time and does not allow for integration of general environmental effects, which often dramatically affect the response \cite{mukamel,franco_deco}.\\

In this work, we introduce a \gls{GQPE} framework as an extension of \gls{QPE} to multi-variate expectations such as those of higher-order correlation functions. \gls{GQPE} provides a general method for obtaining response functions of arbitrary order on a quantum computer such as those used to model nonlinear spectroscopy. Furthermore, \gls{GQPE} enables the incorporation of different spectral broadenings that appear due to the interaction of dissipative environments. The \gls{GQPE} circuit has an intuitive construction that is linked with the physical process of interest, and can directly sample frequencies from the distribution that would be obtained experimentally. Lastly, we extend Lin and Tong's technique \cite{lin_and_tong} for single-ancilla implementation of \gls{QPE} to our new framework to provide a more hardware-friendly implementation.\\

This work is organized as follows. \cref{sec:response} gives a review of linear and nonlinear response theory and discusses a few spectroscopic applications. In \cref{sec:GQPE} we extend the traditional \gls{QPE} framework by introducing \gls{GQPE} and discussing computation of response quantities as an example application. We then further generalize Lin and Tong's approach from \cite{lin_and_tong} to our new framework for a more hardware-friendly implementation of \gls{GQPE}.

\section{Theoretical background: response theory} \label{sec:response}
We start by introducing linear response and highlight its main components, namely correlation functions and spectral broadenings. The generalization to nonlinear processes is then presented.

\subsection{Linear response}
The objective of response theory is to compute how a quantum system described by the Hamiltonian $\hat H = \sum_{j} \lambda_j \ket{\lambda_j}\bra{\lambda_j}$ changes due to an external perturbation. More concretely, response theory offers a way to compute the variation of the expectation of an observable $\hat A$ when the system is perturbed with the operator $\hat V$. Here, we consider a scaled Hamiltonian such that its spectral norm $\|\hat H\| \leq \pi$. We consider the system to start at $t=-\infty$ in an equilibrium state $\hat \rho_0 = \sum_n \rho_n \ket{\lambda_n}\bra{\lambda_n}$ such that $[\hat H,\hat\rho_0]=0$. Examples of such a state would be the ground state, an excited state, or a thermal state. Given an envelope function $f$, the expectation value of $\hat A$ after time-evolution by the perturbed Hamiltonian $\hat{\mathcal{H}}(t')= \hat H + f(t')\hat V$ for time $t$ can be written as a series
\begin{equation} \label{eq:expectation_series}
    \langle \hat A \rangle_t = \sum_{D=0}^{\infty} \langle \hat A^{(D)} \rangle_t,
\end{equation}
where the zeroth-order contribution corresponds to the expectation value of $\hat A$ at $t=0$,
\begin{equation}
    \langle \hat A^{(0)} \rangle_t = \tr{e^{i\hat H t}\hat A e^{-i\hat H t} \hat \rho_0} = \tr{\hat A \hat\rho_0},
\end{equation}
and the first-order correction is written as
\begin{equation}
\langle \hat A^{(1)}\rangle_t = i\int_{0}^\infty dt_1 f(t-t_1) \tr{[\hat A_I(t_1), \hat V]\hat\rho_0}.
\end{equation}
This corresponds to a convolution between the envelope function $f$ and the linear response function
\begin{align}
    \chi^{(1)}_{AV}(t) &= \theta(t)\tr{[\hat A_I(t),\hat V]\hat\rho_0}, \label{eq:linear_response_time}
\end{align}
where $\theta(t)$ is the Heaviside step function and $\hat A_I(t)= e^{i\hat H t}\hat A e^{-i\hat H t}$ is the operator $\hat A$ in the interaction picture representation. Through the use of the convolution theorem, first-order contribution to the expectation value can then be calculated as
\begin{equation} \label{eq:first_correction}
    \langle \hat A^{(1)} \rangle_t = \frac{i}{2\pi}\int_{-\infty}^\infty d\omega \chi_{AV}^{(1)}\left(\omega\right) F(\omega) e^{it\omega},
\end{equation}
where we have defined the Fourier transform of the envelope function $F(\omega)=\ft{f(t)}$. We can expand the frequency-dependent linear response as
\begin{align}
    \chi_{AV}^{(1)}(\omega) &= \ft{\chi_{AV}^{(1)}(t)} \\
    &= \Theta(\omega) * \left(C_{AV}(\omega) + C_{AV}^*(-\omega) \right), \label{eq:linear_response_conv}
\end{align}
where $*$ corresponds to a convolution, $\Theta(\omega) = \ft{\theta(t)}$, and the correlation function
\begin{align}
C_{AV}(\omega) &= \ft{\tr{\hat A_I(t)\hat V \hat\rho_0}} \\
&= 2\pi \sum_{n_{0}n_{1}} \rho_{n_0} A_{n_{0}n_{1}}V_{n_{1}n_{0}} \delta(\Delta_{10} - \omega).
\end{align}
Here we have defined the matrix elements of operators $A_{n_{0}n_{1}}= \bra{\lambda_{n_0}}\hat A\ket{\lambda_{n_1}}$, and the energy differences $\Delta_{10}= \lambda_{n_0} - \lambda_{n_1}$. The Heaviside function enforces causality and is related to the Kramers-Kronig relations \cite{kk1,kk2,mukamel,boyd}. A more in depth discussion of this point appears in \cref{app:lineshapes}. \\

In physical systems, dissipative mechanisms often lead to a spectral broadening. A large class of dissipative effects can be added to the above formalism by making the modification $e^{\pm i\hat H t}\rightarrow e^{\pm i\hat H t}e^{-\hat\Gamma(t)}$ where $[\hat H , \hat\Gamma(t)] = 0$. This enforces a dissipation rate $\gamma_n(t)$ at time $t$ for each eigenstate $\ket{\lambda_n}$. In the case of a scalar uniform broadening, we have $\hat\Gamma(t) = \gamma(t)$, which corresponds to having the same rate for all states. Within our framework, this corresponds to replacing $\Theta(\omega)$ by the line shape $\mathcal{L}(\omega) = \ft{\theta(t)e^{-\gamma(t)}}$ in \cref{eq:linear_response_conv}, obtaining the broadened linear response
\begin{equation}
    \chi_{AV}^{(1)}(\omega) = \mathcal{L}(\omega) * \left(C_{AV}(\omega) + C_{AV}^*(-\omega) \right).\label{eq:broadened_lr}
\end{equation}
Lorentzian, Gaussian, and Voigt functions are commonly used to broaden the spectral lines. In general, arbitrary environmental interactions could be incorporated here by replacing the unitary Hamiltonian-based evolution by a non-unitary Liouvillian evolution that models the desired environmental effects. Several approaches have been proposed for simulating these non-unitary dynamics on a quantum computer \cite{open1,open2,open3,open4,open5,open6}. For simplicity, for the remainder of this work we only consider broadening effects that can be added analytically to the unitary Hamiltonian simulation. A more in-depth discussion of environmental interactions and spectral broadenings is presented in \cref{app:lineshapes}. \\

Upon examination of \cref{eq:broadened_lr} we can see that the response function is composed of two parts: a line shape function $\mathcal{L}(\omega)$, which is convoluted with a sum of correlation functions. Each one of these correlation functions can be associated to a different path taken by the density matrix in Liouville space and a different physical process, which is often represented by a double-sided Feynman diagram \cite{mukamel}. This structure will also be present in nonlinear responses, and is instrumental for the quantum algorithm presented in the next section. \\

A key application of linear response theory is the calculation of the absorption spectrum of quantum systems within the dipole approximation. The absorption spectrum is given by the broadened frequency-domain correlation function between the dipole operator $\hat \mu$ and itself ($\hat A = \hat V = \hat\mu$), often referred to as the auto-correlation
\begin{equation}
    \chi_{\mu\mu}^{(1)}(\omega) = \sum_{n_0n_1} \rho_{n_0}|\mu_{n_1n_0}|^2 \,\mathcal{L}(\Delta_{10}-\omega).
\end{equation}

\subsection{Nonlinear response}
Having shown how to calculate linear responses, we now present the extension to nonlinear responses. In analogy with Eq.~\eqref{eq:first_correction}, we can write the $D^{th}$-order correction to the expectation value
\begin{equation}
    \langle \hat A^{(D)} \rangle_t = i^D \int_{-\infty}^\infty dt_1 \int_{-\infty}^\infty dt_2 ... \int_{-\infty}^\infty dt_D \, \chi_{AV}^{(D)}(t_1,\cdots,t_D)\prod_{j=1}^D f\left(t-\sum_{k=0}^{j-1} t_{D-k}\right),
\end{equation}
which corresponds to a total of $D$ convolutions. We have defined the time-dependent nonlinear response 
\begin{align}
    \chi_{AV}^{(D)}(t_1,...,t_D) = \left(\prod_{j=1}^D \theta(t_j)\right) \,\tr{\left[...\left[\left[\hat A_I\left(\sum_{i=1}^D t_i\right), \hat V_I\left(\sum_{i=1}^{D-1} t_i\right)\right], \hat V_I\left(\sum_{i=1}^{D-2} t_i\right)\right],..., \hat V_I\left(0\right)\right], \hat\rho_0}. \label{eq:commutator_response}
\end{align}
The trace over the nested commutators will give a sum of correlation functions for different permutations of $\hat A_I$ and $\hat V_I$'s with their associated times. The associated multi-dimensional Fourier conjugate then becomes
\begin{equation}
    \chi_{AV}^{(D)}\left(\omega_1, \omega_1+\omega_2,...,\omega_1+...+\omega_D\right) = \int_{-\infty}^{\infty} dt_1 e^{-i\omega_1t_1} \int_{-\infty}^{\infty} dt_2 e^{-i\omega_2t_2} ... \int_{-\infty}^{\infty} dt_D e^{-i\omega_Dt_D} \chi_{AV}^{(D)}(t_1,...,t_D).
\end{equation}
Separating each of the contributions from \cref{eq:commutator_response} we can express the nonlinear response as
\begin{equation}
   \chi_{AV}^{(D)}\left(\vec\omega\right) = \sum_{\alpha=1}^{2^D} R_{AV}^{(D,\alpha)}(\vec\omega),
\end{equation}
where each of the $\alpha=1,...,2^D$ contributions coming from the $D$ commutators can be associated to a different path in Liouville space, which yields a $R^{(D,\alpha)}_{AV}(\vec\omega)$ path response. Note that paths will always appear in pairs with their complex conjugates, having an interference for which only the real/imaginary part survives for $D$ odd/even. Through the use of the convolution theorem, the $D^{th}$-order contribution to the expectation value can then be calculated as
\begin{equation}
    \langle \hat A^{(D)} \rangle_t = \left(\frac{i}{2\pi}\right)^D\int_{-\infty}^\infty d\omega_1\int_{-\infty}^\infty d\omega_2...\int_{-\infty}^\infty d\omega_D \chi_{AV}^{(D)}\left(\omega_1, \omega_1+\omega_2,...,\omega_1+...+\omega_D\right) \prod_{j=1}^D F(\omega_j)\,e^{it\omega_j}.
\end{equation}
We now explicitly write the path response corresponding to the first term where all commutators were taken with the positive sign:
\begin{align}
    R_{AV}^{(D,1)}(\vec\omega) &= \int_{-\infty}^\infty dt_1 \int_{-\infty}^\infty dt_2 ... \int_{-\infty}^\infty dt_D e^{-i\vec\omega\cdot\vec t} \theta( t_1)...\theta( t_D) \textrm{Tr}\Bigg\{e^{i\hat H ( t_1+...+ t_D)} \hat A e^{-i\hat H  t_D} \hat V e^{-i\hat H  t_{D-1}}\cdots\hat V e^{-i\hat H  t_1} \hat V \hat\rho_0\Bigg\} \nonumber\\
    &= \sum_{n_0,n_1,...,n_D} \rho_{n_0}A_{n_0n_1} V_{n_1 n_2} ... V_{n_{D-1} n_D} V_{n_D n_0} \Theta(\Delta_{10}-\omega_1) \Theta(\Delta_{20}-\omega_2) ... \Theta(\Delta_{D0}-\omega_D).
\end{align}
All other paths will have a completely analogous treatment, while having a permutation of the $\hat A_I$ and $\hat V_I$ operators. Including the line shape into each path response, we can write
\begin{equation} \label{eq:path_as_convolution}
    R_{AV}^{(D,\alpha)}(\vec\omega) = \left(\prod_{i=1}^D \mathcal{L}(\omega_i)\right) *^D C^{(D,\alpha)}(\vec\omega),
\end{equation}
where $*^D$ is the convolution over each frequency $\omega_1,\cdots,\omega_D$, and $C^{(D,\alpha)}(\vec\omega)$ is the corresponding correlation function for the path $\alpha$. We find $C^{(D,\alpha)}(\vec\omega)$ by writing the associated expectation value, e.g. for $D=3$ and a particular path $\alpha_0$:
\begin{equation}
    C^{(3,\alpha_0)}(\vec t) = \tr{\hat V_I(0)\hat V_I( t_1)\hat A_I( t_1+ t_2+ t_3) \hat V_I( t_1+ t_2) \hat\rho_0},
\end{equation} 
after which we expand the time-evolutions and cancel trivial terms
\begin{align}
    C^{(3,\alpha_0)}(\vec t) &= \tr{\hat V e^{i\hat H  t_1} \hat V e^{-i\hat H  t_1}e^{i\hat H ( t_1+ t_2+ t_3)}\hat A e^{-i\hat H ( t_1+ t_2+ t_3)} e^{i\hat H ( t_1+ t_2)} \hat V e^{-i\hat H ( t_1+ t_2)} \hat\rho_0} \\
&= \tr{\hat V e^{i\hat H  t_1} \hat V e^{i\hat H ( t_2+ t_3)}\hat A e^{-i\hat H  t_3} \hat V e^{-i\hat H ( t_1+ t_2)} \hat\rho_0}. \label{eq:to_be_referenced}
\end{align}
As an application example from nonlinear spectroscopy we consider a resonant Raman process, for which only a single pair of conjugate paths in Liouville space has a significant contribution. Considering an incoming/scattered photon with frequency $\omega_I/\omega_S$, the resonant Raman cross section can be expressed through the Kramers-Heisenberg formula \cite{kk2}:
\begin{align}
    \frac{d^2\sigma(\omega_I,\omega_S)}{d\omega_Id\Omega} &\propto \omega_I \omega_S^3 \chi_{\mu\mu}^{(3)}(\omega_I,\omega_S-\omega_I,\omega_I) \\
    \chi_{\mu\mu}^{(3)}(\omega_I,\omega_S-\omega_I,\omega_I) &= \operatorname{Re} \left[\sum_{n_0,n_1,n_2,n_3} \mu_{n_0n_1}\mu_{n_1n_2}\mu_{n_2n_3}\mu_{n_3n_0} \mathcal{L}_{\eta_{int}}(\Delta_{10}-\omega_I)\mathcal{L}^*_{\eta_{int}}(\Delta_{30}-\omega_I)\mathcal{L}_{\eta_f}(\Delta_{20}-\omega_I+\omega_S) \right]. \label{eq:raman}
\end{align}
For simplicity we assumed real dipole matrix elements. Note how different broadenings appear: this is due to the energy-dependent dissipation $\hat\Gamma(t)$. The resonant nature of this process makes intermediate states localized in energies, meaning the dissipation at each different stage can be approximated by an associated constant rate, resulting in the broadenings $\eta_{int}$ and $\eta_f$.

\section{Generalized quantum phase estimation} \label{sec:GQPE}

In this section we introduce a generalization of the \gls{QPE} algorithm and look at nonlinear response calculations as an example application. We first extend the scope of traditional \gls{QPE} to include not only phase estimation of an evolution operator $e^{i\hat H}$ but also phase estimation on interaction picture operators $\hat A_I(t)\hat B = e^{i\hat H t} \hat A e^{-i\hat H t}B$. We then further generalize the framework to perform multi-variate phase estimation on a multi-variable operator such as $\hat A_I(t_1)\hat B_I(t_2)\hat C_I(t_3)\hat D$. Given an initial eigenstate of the system $\ket{\lambda_{n_0}}$, the goal of QPE can be interpreted as estimating the function
\begin{align}
    C(\omega) &= \frac{1}{2\pi}\int_{-\infty}^\infty dt\, e^{-i\omega t}\bra{\lambda_{n_0}}e^{i\hat Ht}\ket{\lambda_{n_0}}\nonumber\\
    &= \delta(\lambda_{n_0}-\omega).
\end{align}
We can similarly extend the scope to interaction picture operators
\begin{align}
    C(\omega) &= \frac{1}{2\pi}\int_{-\infty}^\infty dt\, e^{-i\omega t}\bra{\lambda_{n_0}}e^{i\hat H t} \hat A e^{-i\hat H t}\hat B\ket{\lambda_{n_0}}\nonumber\\
    &= \sum_{n_1} A_{n_0n_1}B_{n_1n_0}\, \delta(\Delta_{10}-\omega).
\end{align}
Here we have written $A_{n_0n_1} = \bra{\lambda_{n_0}} \hat A\ket{\lambda_{n_1}}$, and $\Delta_{10} = \lambda_{n_0} - \lambda_{n_1}$. Given a finite energy resolution $\eta\in(0,1)$, we relax the problem by replacing $\delta(\omega)$ with a window function $\mathcal{L}_\eta(\omega)$ represented as a truncated Fourier series $\mathcal{L}_{\eta}(\omega) = \frac{1}{\sqrt{N}}\sum_{k=0}^{N-1}\alpha_k e^{ik\omega}$ with $\| \vec{\alpha}\|^2 = 1$, such that for $\omega\in[-\eta, \eta]$, $|\mathcal{L}_\eta(\omega)|^2\in \mathcal{O}(1)$. A typical example of this is the sinc function where all $\alpha_k = \frac{1}{\sqrt{N}}$ for $N\in \mathcal{O}(1/\eta)$, prepared by applying a Hadamard on each of the wires in the time register.\\

\begin{figure}[t!]
\centering
\begin{minipage}{1\textwidth}
\[
\Qcircuit @C=1.5em @R=1em {
     & \gate{U_{\mathcal{L}}} &  \qw \barrier{4}& \qw &\gate{} \qwx[1]  & \qw \barrier{4} &\qw &\gate{\text{QFT}^\dag} & \qw \\
     & \gate{U_{\mathcal{L}}} & \qw &\qw &\gate{} \qwx[2]  & \qw &\qw& \gate{\text{QFT}^\dag} & \qw \\
    \lstick{\hat U_{R(\vec\omega)}:=}   & \vdots                      & & &     &  & & \vdots                    &     \\
     & \gate{U_{\mathcal{L}}} & \qw &\qw & \gate{} \qwx[1] & \qw &\qw&\gate{\text{QFT}^\dag} & \qw \\
     & \gate{U_{\lambda_{n_0}}}  & \qw & \qw& \gate{\hat V_I^{(0)}(t_D)\hat V_I^{(1)}(t_{D-1})\cdots \hat V_I^{(D-1)}(t_1) \hat V^{(D)}} & \qw &\qw&\gate{U_{\lambda_{n_0}}^\dag} & \qw
}
\]
(a)

\[
\underbrace{\Qcircuit @C=0.5em @R=1em {
    & \qw & \qw & \qw& \qw & \qw & \qw & \qw& \qw & \qw & \qw & \qw& \qw &  \ctrl{9} & \qw & \qw & \qw & \qw & \qw & \qw & \qw & \qw & \qw & \ctrl{9} & \qw \\
    & \qw  & \qw & \qw& \qw & \qw & \qw & \qw& \qw & \qw & \qw & \qw& \qw &  \qw & \ctrl{8} & \qw & \qw & \qw & \qw & \qw & \qw & \qw & \ctrl{8} &\qw &  \qw\\
    & \vdots &  &  &  & &&  &  & &  &  & &  &  &  & \ddots & &  & &  & \Ddots \\
    & \qw & \qw  & \qw& \qw & \qw & \qw & \qw& \qw & \qw & \qw & \qw& \qw & \qw & \qw & \qw  & \qw & \ctrl{6} & \qw & \ctrl{6}  & \qw & \qw & \qw & \qw  & \qw
    \inputgroupv{1}{4}{1em}{1.5em}{\ket{t_2}}\\
    &  &  &  & &  &  & &  &  & & &  & &  & &  &  & \\
    & \qw &  \ctrl{4} & \qw & \qw & \qw & \qw & \qw & \qw & \qw & \qw & \qw & \ctrl{4} & \qw & \qw & \qw & \qw & \qw & \qw & \qw & \qw & \qw & \qw & \qw & \qw \\
    & \qw &  \qw & \ctrl{3} & \qw & \qw & \qw & \qw & \qw & \qw & \qw & \ctrl{3} &\qw &  \qw & \qw & \qw & \qw & \qw & \qw & \qw & \qw & \qw & \qw & \qw & \qw\\
    &  \vdots &  &  & \ddots & &  & &  & \Ddots \\
    & \qw & \qw & \qw & \qw  & \qw & \ctrl{1} & \qw & \ctrl{1}  & \qw & \qw & \qw & \qw   & \qw & \qw & \qw & \qw & \qw & \qw & \qw & \qw & \qw & \qw & \qw & \qw
    \inputgroupv{6}{9}{1em}{1.5em}{\ket{t_1}}\\
    \lstick{\ket{\lambda_{n_0}}}& \gate{V^{(2)}} & \gate{U} & \gate{U^2} & \cdots & & \gate{U^{\frac{N}{2}}} & \gate{V^{(1)}} & \gate{U^{\frac{-N}{2}}} & \cdots & &  \gate{U^{-2}} & \gate{U^{-1}}  & \gate{U} & \gate{U^2} & \cdots & & \gate{U^{\frac{N}{2}}} & \gate{V^{(0)}} & \gate{U^{\frac{-N}{2}}} & \cdots & &  \gate{U^{-2}} & \gate{U^{-1}}\\
    &
}}_{\hspace{-1cm}\Qcircuit @C=2.5em @R=1em {
    &\\
     &\lstick{\ket{t_2}} &\gate{} \qwx[1] & \qw \\
     &\lstick{\ket{t_1}} &\gate{} \qwx[1] & \qw \\
     &\lstick{\ket{\lambda_{n_0}}} & \gate{\hat V^{(0)}_I(t_2)\hat V^{(1)}_I(t_1)\hat V^{(2)}} & \qw
}}
\]
(b)
\caption{(a) Quantum circuit for \gls{GQPE} on the multi-variate operator $\hat V^{(0)}_I(t_D)\hat V^{(1)}_I(t_{D-1})\cdots \hat V^{(D-1)}_I(t_1) \hat V^{(D)}$. It starts by preparing the system register $\ket{\lambda_{n_0}}=\hat U_{\lambda_{n_0}}\ket{0}$ and each of the time registers $\sum_{k=0}^{N-1} \alpha_k \ket{k} = \hat U_{\mathcal{L}}\ket{0}$ for a window function $\mathcal{L}(\omega) = \frac{1}{\sqrt{N}}\sum_{k=0}^{N-1}\alpha_k e^{ik\omega}$. Next, we apply the multiplexed time correlation operator between operators $\hat V^{(0)},\hat V^{(1)},\cdots \hat V^{(D)}$ followed by a Fourier transform of the time registers and unpreparing of the system register. (b) Expanded circuit implementation of a two-variable operator $\hat V^{(0)}_I(t_2)\hat V^{(1)}_I(t_1)\hat V^{(2)} = e^{i\hat Ht_2}\hat V^{(0)}e^{-i\hat Ht_2}e^{i\hat Ht_1}\hat V^{(1)}e^{-i\hat Ht_1}\hat V^{(2)}$, where we have written $\hat U = e^{-i\hat H}$. Since $\hat V^{(i)}$s are generally not unitary, their action corresponds to applying a block-encoding of these operators.}
\vspace{-0.3cm}
\label{fig:corr_func}
\end{minipage}
\end{figure}

The extension to a multi-variable quantum phase estimation is similar. We consider some $D$-variable operator $\hat V^{(0)}_I(t_D)\hat V^{(1)}_I(t_{D-1})\cdots \hat V^{(D-1)}_I(t_1) \hat V^{(D)}$ and write
\begin{align}
    C(\vec\omega) &= \frac{1}{(2\pi)^D}\int_{-\infty}^\infty dt_1\cdots \int_{-\infty}^\infty dt_D\, e^{-i\vec\omega\cdot \vec t}\bra{\lambda_{n_0}}\hat V^{(0)}_I(t_D)\hat V^{(1)}_I(t_{D-1})\cdots \hat V^{(D-1)}_I(t_1) \hat V^{(D)}\ket{\lambda_{n_0}}\nonumber\\
    &= \sum_{n_1,...,n_D} V^{(D)}_{n_Dn_0}\prod_{j=1}^{D} V^{(j-1)}_{n_{j-1}n_{j}}\, \delta(\Delta_{(j-1)j}-\omega_{D-j+1}).
\end{align}
Given a finite energy resolution $\eta\in(0,1)$, we can again relax the problem by replacing $\delta(\omega)$ with a window function $\mathcal{L}_\eta(\omega)$ represented as a truncated Fourier series $\mathcal{L}_{\eta}(\omega) = \frac{1}{\sqrt{N}}\sum_{k=0}^{N-1}\alpha_k e^{ik\omega}$ with $\| \vec{\alpha}\|^2 = 1$, obtaining the broadened correlation function
\begin{align}\label{eq:quant_of_interest}
    R(\vec\omega) = \sum_{n_1,...,n_D} V^{(D)}_{n_Dn_0}\prod_{j=1}^{D} V^{(j-1)}_{n_{j-1}n_{j}}\, \mathcal{L}_\eta(\Delta_{(j-1)j}-\omega_{D-j+1}).
\end{align}
However, note that in order to be able to estimate the multi-variate correlation function at any given $\vec{\omega}\in[-\eta, \eta]^D$ with accuracy $\epsilon$, we need $\prod_{j=1}^d |\mathcal{L}_{\eta}(\omega_j)|^2 \in 1-\mathcal{O}(\epsilon)$. This means that for each $\omega\in[-\eta, \eta]$ component we need
\begin{equation} \label{eq:window_requirement}
    |\mathcal{L}_{\eta}(\omega)|^2 \sim 1-\mathcal{O}\left(\frac{\epsilon}{D}\right),
\end{equation}
which can be done most efficiently using a Kaiser window \cite{berry2024analyzing,vs_qsvt}, leading to $N \in \mathcal{O}(\log(D/\epsilon)/\eta)$. The circuit implementation of GQPE is depicted in \cref{fig:corr_func}. We describe the algorithm in full detail below.\\

\newpage
\noindent\textbf{Generalized quantum phase estimation circuit}

\begin{enumerate}
    \item Allocate a register for the system as well as $D$ extra $(\log N)$-qubit time registers.
    \item Prepare the system register $\ket{\lambda_{n_0}}=\hat U_{\lambda_{n_0}}\ket{0}$ and each of the time registers $\sum_{k=0}^{N-1} \alpha_k \ket{k} = \hat U_{\mathcal{L}}\ket{0}$ for a window function $\mathcal{L}(\omega) = \frac{1}{\sqrt{N}}\sum_{k=0}^{N-1}\alpha_k e^{ik\omega}$. We refer the reader to \cite{ini_state, dynamical_cooling} for a more in-depth discussions on initial state preparation.
    \item Apply the multiplexed $D$-variable operator of interest to the combined space of all registers (e.g. $\hat A_I(t_2)\hat B_I(t_1)\hat C$ as shown in \cref{fig:corr_func}).
    \item Apply $\text{QFT}^\dag$ to each time register and $\hat U_{\lambda_{n_0}}^\dag$ to the system register.
\end{enumerate}
It results in the unitary $\hat U_{R(\vec\omega)}$ which encodes \cref{eq:quant_of_interest} (see \cref{app:proof}).
\begin{align}\label{eq:main_unitary}
    R\left(\frac{\vec \omega \cdot 2\pi}{N}\right)  = \bra{\vec \omega, 0}\hat U_{R(\vec\omega)}\ket{0, 0}.
\end{align}
We now show two ways to use this circuit to estimate $R\left(\vec \omega\right)$.\\ 

\noindent\textbf{Estimation protocols}

\begin{enumerate}
    \item Utilize $\hat U_{\vec{\omega}} \ket{\Vec{\omega}, 0} = \ket{0}$ consisting of single qubit $X$ gates to re-write \cref{eq:main_unitary} as an expectation value $\bra{0}\hat U_{\vec{\omega}}\hat U_{R(\vec\omega)}\ket{0}$ and estimate $R\left(\vec \omega\right)$ up to accuracy $\epsilon$ for a particular $\vec{\omega}$ with an overhead of $\mathcal{O}(1/\epsilon)$ via amplitude estimation \cite{qae}.
    \item Model each frequency $\vec{\omega}$ as a Bernoulli variable with probability of success given by $\left|R\left(\vec \omega\right)\right|^2$ and sample the circuit in \cref{fig:corr_func}, counting each sampled $\vec{\omega}$ not only as a success for its own Bernoulli variable but also as a failure for all other frequencies. This allows us to estimate $\left|R\left(\vec \omega\right)\right|^2$ up to accuracy $\epsilon^2$ and thus $\left|R\left(\vec \omega\right)\right|$ up to accuracy $\epsilon$ at all points with an overhead of $\mathcal{O}(1/\epsilon^4)$.
\end{enumerate}
Note that for the second approach, in most cases it suffices to estimate $\left|R\left(\vec \omega\right)\right|$ at low accuracy to find peaks in the distribution to be further resolved by approach one. However, it is also possible to estimate $R\left(\vec \omega\right)$ directly via the second approach by separating the real and imaginary parts of the function as shown in \cref{app:proof}. We now describe the complexity of the algorithm.
\begin{theorem}[Complexity] \label{theo:freq_comp}
    Given a system's time-evolution operator $e^{i\hat H}$ and set of unitaries $\{\hat U_{V^{(j)}}\}_{j=0}^D$ that block-encode operators $\hat V^{(j)}/|\hat V^{(j)}|_1$,  we can estimate $R\left(\vec \omega\right)$ from \cref{eq:quant_of_interest} with energy resolution $\eta \in (0,1)$ up to accuracy $\epsilon \in (0,1)$ by either of the following:
    \begin{enumerate}
        \item At a particular $\vec{\omega}$ using $\tilde{\mathcal{O}}\left(\frac{D\cdot \Omega}{\eta\cdot \epsilon}\right)$ queries to $e^{\pm i\hat H}$, plus $\mathcal{O}\left(\frac{\Omega}{\epsilon}\right)$ queries to each block-encoding unitary in $\{\hat U_{V^{(j)}}\}_{j=0}^D$ as well as $\hat U_{\lambda_n}$ for a total of $\mathcal{O}(1)$ samples. \\ \\
        In total, this procedure requires $\tilde{\mathcal{O}}\left(\frac{D\cdot \Omega}{\eta\cdot \epsilon}\right)$ queries to the time-evolution oracles and $\mathcal{O}\left(\frac{\Omega}{\epsilon}\right)$ queries to the block-encoding and initial state preparation oracles.
        \item At every point $\vec{\omega}$ simultaneously. Each circuit uses $\tilde{\mathcal{O}}\left(\frac{D}{\eta}\right)$ queries to $e^{\pm i\hat H}$, plus $\mathcal{O}(1)$ queries to each block-encoding unitary as well as $\hat U_{\lambda_n}$. The associated sample complexity is $\mathcal{O}\left(\frac{\Omega^4}{\epsilon^4}\right)$. \\ \\
        In total, this procedure requires $\tilde{\mathcal{O}}\left(\frac{D\cdot\Omega^4}{\eta\cdot \epsilon^4}\right)$ calls to the time-evolution oracles alongside $\mathcal{O}\left(\frac{\Omega^4}{\epsilon^4}\right)$ calls to the block-encoding and initial state preparation oracles. 
    \end{enumerate}
    Here we have defined $\Omega = \prod_{j=0}^D |\hat V^{(j)}|_1$ as the product of 1-norms which appears when normalizing the associated block-encodings.
\end{theorem}

\begin{proof}
This proof hinges on the action of the $\hat U_{R(\vec\omega)}$ circuit shown in \cref{fig:corr_func}. As shown in \cref{app:proof}, up to some constants related to the binary representation of the frequencies in $\vec\omega$, this unitary acts on the all-zeros computational basis state as $\hat U_{R(\vec\omega)} \ket{0,0} = \Omega^{-1} \sum_{\vec\omega} R\left(\vec\omega\right) \ket{\vec\omega,0} + \ket{\perp}$, where $\ket{\perp}$ corresponds to the system register being in a state that is different from $\ket{0}$ and we have included the $\Omega$ factor appearing from the block-encodings. We now show how the complexities of the estimation protocols correspond to those stated in \cref{theo:freq_comp}.
\begin{enumerate}
    \item We start by writing the target response function as the expectation value $R(\vec\omega) = \Omega \bra{0}\hat U_{\vec\omega} \hat U_{R(\vec\omega)}\ket{0}$. Determining the expectation value to accuracy $\epsilon$ can be optimally done through the use of amplitude estimation routines \cite{qae}, having that each circuit requires $\mathcal{O}(1/\epsilon)$ calls to the target unitary and $\mathcal{O}(1)$ samples. Since $\Omega$ appears multiplying the expectation value, the accuracy needs to be adjusted as $\epsilon\rightarrow\epsilon/\Omega$ as to determine $R(\vec\omega)$ to accuracy $\epsilon$. Noting that the maximum evolution time required for resolving peaks with a width of $\eta$ is $\tilde{\mathcal{O}}(1/\eta)$ for each of the $D$ time/frequency registers, we thus get a circuit complexity of $\tilde{\mathcal{O}}(D\cdot \Omega/\eta\cdot\epsilon)$ calls to the time-evolution $e^{\pm i\hat H}$, and $\mathcal{O}(\Omega/\epsilon)$ queries to initial state preparation oracle and each of the block-encoding unitaries.
    \item We start by noting that the probability of measuring the output state $\ket{\vec\omega,0}$ corresponds to $P(\vec\omega,0) = \Omega^{-2} |R(\vec\omega)|^2$. If $|R(\vec\omega)|$ is to be determined to accuracy $\epsilon$, we then need to recover $P(\vec\omega,0)$ to accuracy $\epsilon^2/\Omega^2$, which requires $\mathcal{O}(\Omega^4/\epsilon^4)$ circuit repetitions. Each $\hat U_{R(\vec\omega)}$ will require $\tilde{\mathcal{O}}(D/\eta)$ calls to the time-evolution oracles, alongside $\mathcal{O}(1)$ calls to each of the block-encoding and $\hat U_{\lambda_n}$ unitaries. Note that this procedure will then determine $|R(\vec\omega)|$ to the target accuracy $\epsilon$. Recovering $R(\vec\omega)$ can then be done as shown in \cref{app:proof} without increasing the complexity.
\end{enumerate}
\end{proof}
We note that in general, the block-encoding of an operator might not always be successful due to the non-unitarity of the block-encoded operator and the fact that the block-encoded operator will be divided by the 1-norm of the coefficients in its decomposition into a linear combination of unitaries. However, we do not need to pay a penalty for failures due to the first reason as they are part of the quantity being estimated, that is the non-unitarity of the perturbations is also part of $R(\vec\omega)$. This failure probability is already considered by the fact that we recover the target response up to the $\Omega$ factor. Thus, having a failed application of any of the block-encodings would reflect having a distribution such that its norm $\mathcal{N}= \Omega^{-2} \sum_{\vec\omega} |R(\vec\omega)|^2$ is smaller than $1$. To make this point clearer, in the case where $R(\vec\omega)$ is trivially $0$ over all $\vec\omega$'s, we would have a block-encoding failure over each run of the circuit, each one containing information of the associated $\mathcal{N}=0$.

\subsection{Complete square simplification}
\begin{figure}
    \centering
    \[
\Qcircuit @C=1.5em @R=1em {
     \lstick{\ket{0}} & \gate{U_{\mathcal{L}}} &  \qw & \qw &\gate{} \qwx[1]  & \qw &\qw &\gate{\text{QFT}^\dag} & \meter & \rstick{\omega_1} \\
     \lstick{\ket{0}} & \gate{U_{\mathcal{L}}} & \qw &\qw &\gate{} \qwx[2]  & \qw &\qw& \gate{\text{QFT}^\dag} & \meter & \rstick{\omega_2} \\
        & \vdots                      & & &     &  & & \vdots                    &     \\
     \lstick{\ket{0}} & \gate{U_{\mathcal{L}}} & \qw &\qw & \gate{} \qwx[1] & \qw &\qw&\gate{\text{QFT}^\dag} & \meter & \rstick{\omega_D} \\
     \lstick{\ket{0}} & \gate{U_{\lambda_{n_0}}}  & \qw & \qw& \gate{\hat V_I^{(0)}(t_D)\hat V_I^{(1)}(t_{D-1})\cdots \hat V_I^{(D-1)}(t_1) \hat V^{(D)}} & \qw &\qw & \qw & \qw
}
\]
    \caption{Quantum circuit for \gls{GQPE} sampling frequencies from the response function with complete-square structure shown in \cref{eq:complete_square}. Note that this is associated to a response with order $2D+1$.}
    \label{fig:complete_square}
\end{figure}

We now introduce an algorithmic simplification for the case where the quantity of interest can be written as a complete square, namely 
\begin{equation} \label{eq:complete_square}
    R^{(2)}(\vec\omega) = \left|\left| \sum_{n_1,\cdots,n_D} \prod_{j=0}^D \mathcal{L}(\Delta_{(j-1)j} -\omega_j) V^{(j)}_{n_{j-1}n_j} \ket{\lambda_{n_D}}\right|\right|^2.
\end{equation}
Examples of such processes include linear absorption and emission, alongside resonant Raman scattering within the dipole approximation \cite{mukamel}. Note that such cases require that at least one of the perturbations is also the observable of interest, while all other perturbations act twice. The simplified circuit is shown in \cref{fig:complete_square}. This corresponds to a response of order $2D+1$, while only $D$ time/frequency registers are required. The probability of measuring $\vec\omega$ in the ancilla registers corresponds to $P(\vec\omega) = \Omega^{-2}R^{(2)}(\vec\omega)$. Note that the $\Omega=\prod_{j=0}^D|\hat V^{(j)}|_1$ defined here runs over $D+1$ elements, which makes it quadratically smaller than the one with $2D+2$ elements that would be used in \cref{theo:freq_comp}. When compared with the more general circuit in \cref{fig:corr_func}, this complete square approach then reduces the number of block-encoding applications from $2D+2$ to $D+1$ and the associated $\Omega$ by a square-root factor.

\begin{theorem}[Complete square complexity] \label{theo:comp_square}
Given a system's time-evolution operator $e^{i\hat H}$ and set of unitaries $\{\hat U_{V^{(j)}}\}_{j=0}^D$ that block-encode operators $\hat V^{(j)}/|\hat V^{(j)}|_1$, we can estimate the (complete square) response $R^{(2)}(\vec\omega)$ with order $2D+1$ from \cref{eq:complete_square} with energy resolution $\eta\in(0,1)$ up to accuracy $\epsilon\in(0,1)$ by either of the following:
\begin{enumerate}
    \item At a particular $\vec\omega$ using $\tilde{\mathcal{O}}\left(\frac{D\cdot \Omega^2}{\eta\cdot\epsilon}\right)$ queries to $e^{\pm i\hat H}$, plus $\mathcal{O}\left(\frac{\Omega^2}{\epsilon}\right)$ queries to each block-encoding unitary in $\{\hat U_{V^{(j)}}\}_{j=0}^D$ as well as $\hat U_{\lambda_n}$ for a total of $\mathcal{O}(1)$ samples. \\ \\
    In total, this procedure requires $\tilde{\mathcal{O}}\left(\frac{D\cdot \Omega^2}{\eta\cdot\epsilon} \right)$ queries to the time-evolution oracles and $\mathcal{O}\left(\frac{\Omega^2}{\epsilon}\right)$ queries to the block-encoding and initial state preparation oracles.
    \item At every point $\vec\omega$ simultaneously. Each circuit uses $\tilde{\mathcal{O}}\left(\frac{D}{\eta}\right)$ queries to $e^{\pm i\hat H}$, plus $\mathcal{O}(1)$ queries to each block-encoding unitary as well as $\hat U_{\lambda_n}$. The associated sample complexity is $\mathcal{O}\left(\frac{\Omega^4}{\epsilon^2}\right)$. \\ \\
    In total, this procedure requires $\tilde{\mathcal{O}}\left(\frac{D\cdot \Omega^4}{\eta\cdot\epsilon^2} \right)$ queries to the time-evolution oracles and $\mathcal{O}\left(\frac{\Omega^4}{\epsilon^2}\right)$ queries to the block-encoding and initial state preparation oracles.
\end{enumerate}

\end{theorem}
\begin{proof}
From the proof of the action of $\hat U_{R(\vec\omega)}$ in \cref{app:proof}, it follows immediately that the circuit shown in \cref{fig:complete_square} before measurement yields the wavefunction $\Omega^{-1}\sum_{\vec\omega}\sum_{n_1,\cdots,n_D}\prod_{j=0}^D \mathcal{L}(\Delta_{(j-1)j} - \omega_j) V_{n_{j-1}n_j}^{(j)} \ket{\vec\omega,\lambda_{n_D}}$. Measurement of $\vec\omega$ in the ancilla registers will have an associated probability of $P(\vec\omega) = \Omega^{-2}R^{(2)}(\vec\omega)$. By determining $P(\vec\omega)$ with accuracy $\epsilon/\Omega^2$ we can then recover $R^{(2)}(\vec\omega)$ with accuracy $\epsilon$. We now show how this can be used to recover the complexities stated in \cref{theo:comp_square}. 
\begin{enumerate}
    \item Determination of the probability $P(\vec\omega)$ to accuracy $\epsilon/\Omega^2$ can be done optimally through the use of amplitude estimation routines \cite{qae}, each circuit needing $\mathcal{O}(\Omega^2/\epsilon)$ calls to the associated unitary, namely the circuit shown in \cref{fig:complete_square} before measurement, and $\mathcal{O}(1)$ samples.  Considering that each unitary requires $\tilde{\mathcal{O}}(D/\eta)$ queries to $e^{\pm i\hat H}$ and $\mathcal{O}(1)$ queries to each block-encoding and $\hat U_{\lambda_n}$, we recover a total cost of $\tilde{\mathcal{O}}(D\cdot \Omega^2/\eta\cdot\epsilon)$ queries to $e^{\pm i\hat H}$, $\mathcal{O}(\Omega^2/\epsilon)$ queries to each block-encoding unitary and $\hat U_{\lambda_n}$, and $\mathcal{O}(1)$ samples.
    \item Recovering $P(\vec\omega)$ to accuracy $\epsilon/\Omega$ can be done by directly sampling the circuit in \cref{fig:complete_square} a total of $\mathcal{O}(\Omega^4/\epsilon^2)$ times. Each circuit requires $\tilde{\mathcal{O}}(D/\eta)$ calls to $e^{\pm i\hat H}$, alongside $\mathcal{O}(1)$ calls to each of the block-encodings and $\hat U_{\lambda_n}$ unitaries.
\end{enumerate}
\end{proof}

\subsection{Application to nonlinear spectroscopy}

A straightforward application of the \gls{GQPE} framework presented above is the calculation of response quantities presented in \cref{sec:response}
\begin{equation} \label{eq:gen_path}
    R(\vec\omega) = \left(\prod_{j=1}^D \mathcal{L}(\omega_j)\right) *^D C(\vec\omega) .
\end{equation}
The central idea is to apply the convolution theorem, obtaining the $D$-dimensional Fourier transform
\begin{equation} 
    R(\vec\omega) = \mathcal{F}^D_{\vec t \rightarrow \vec\omega}\left\{\left(\prod_{j=1}^D \mathcal{F}^{-1}_{\omega_j\rightarrow t_j}\left\{\mathcal{L}(\omega_j)\right\}\right) \cdot C(\vec t)\right\}, \label{eq:td_resp}
\end{equation}
and write the line shape as a scaled and truncated Fourier series $\mathcal{L}(\omega) = \frac{1}{\sqrt{N}}\sum_{k=0}^{N-1}\alpha_k e^{ik\omega}$ with $\| \vec{\alpha}\|^2 = 1$. One would then choose the multi-variate operator for phase estimation to be the associated correlation function $C(\vec t)$ (e.g. $\hat V e^{i\hat H  t_1} \hat V e^{i\hat H ( t_2+ t_3)}\hat A e^{-i\hat H  t_3} \hat V e^{-i\hat H ( t_1+ t_2)}$ from \cref{eq:to_be_referenced}), and run the \gls{GQPE} protocol.\\

However, note that in many cases the broadening has a physical meaning and appears as a result of dissipative environmental effects. Hence, in such cases one might not even require the line shape to satisfy the condition from \cref{eq:window_requirement}. Finally, while in the above we assumed the input is a pure eigenstate $\ket{\lambda_{n_0}}$, the algorithm can be applied to an equilibrium mixed state $\hat\rho_0$ via either a Stinespring dilation-based preparation of $\hat\rho_0$, or by sampling the preparation unitary $\hat U_{\lambda_{n_0}}$ for each run with probability $\tr{\bra{\lambda_{n_0}}\hat\rho_0 \ket{\lambda_{n_0}}}$. Refs.~\cite{thermal1, thermal2} show how thermal states can be prepared on quantum hardware. \\

As an example application, we now show how the scattering amplitude from a resonant Raman process can be obtained using the \gls{GQPE} algorithm. We start by noting that even though this quantity is coming from a third-order response, there are only two frequencies that appear throughout the resonant process. The target quantity can be written as a complete square, meaning we can directly use the complete square protocol outlined above to recover it from a \gls{GQPE} circuit that uses only two time/frequency registers. We start by writing the target quantity in \cref{eq:raman} in its most commonly used form using the Kramers-Heinseberg formula and showcasing its complete square structure:
\begin{align}
    \chi^{(3)}_{\mu\mu}(\omega_I,\omega_S-\omega_I,\omega_I) &= \sum_{n_2} \delta(\Delta_{20} - \omega_I+\omega_S) \left|\sum_{n_1} \mu_{n_2n_1}\mu_{n_1n_0} \mathcal{L}_{\eta_{int}}(\Delta_{10} - \omega_I)\right|^2 \\
    &\approx \left|\left|\sum_{n_1,n_2} \mu_{n_2n_1}\mu_{n_1n_0} \mathcal{L}_{\delta}(\Delta_{20}-\omega_I+\omega_S)\mathcal{L}_{\eta_{int}}(\Delta_{10}-\omega_I)\ket{\lambda_{n_2}}\right|\right|^2,
\end{align}
where we defined the Lorentzian line shape $\mathcal{L}_{\eta_{int}}(\omega) \propto (\omega + i\eta_{int})^{-1}$ alongside the $\mathcal{L}_\delta(\omega)$ line shape such that its square approximates the Dirac delta distribution. Approximating the Dirac delta can be done most efficiently using a Kaiser window. The different line shapes reflect the fact that the highly excited virtual states appearing in this process have an environmental relaxation rate that is significantly faster than that of the lower energy final states. Considering the output of a \gls{GQPE} circuit for estimating a complete square response seen in \cref{fig:complete_square}, namely $R^{(2)}(\omega_1,\omega_2)$, we recover the target quantity as
\begin{equation} \label{eq:gqpe_to_raman}
    \chi^{(3)}_{\mu\mu}(\omega_I,\omega_S-\omega_I,\omega_I) = R^{(2)}(\omega_S+\omega_I+E_{n_0}, \omega_I+E_{n_0}).
\end{equation}

This shows how we can use a \gls{GQPE} circuit to simulate a nonlinear spectroscopy experiment, in this case corresponding to a resonant Raman scattering process. Note that the circuit can be modified by adding a multiplexed application with $e^{i\hat H(t_1+t_2)}$ before the first block-encoding of $\hat\mu$ to include the $E_{n_0}$ factors. However, if this energy is known it can be simply included as a shift as shown in \cref{eq:gqpe_to_raman}, reducing the overall time-evolution implementation cost by a factor of two. In addition, the system register can be prepared from the start with a normalized state corresponding to $\hat\mu\ket{\lambda_{n_0}}$, which would remove one of the block-encodings. This effectively rescales $\Omega$ by a square-root factor, which for complexity purposes changes the $\Omega^4$ factors appearing in \cref{theo:comp_square} to $\Omega^2$. \\
Finally, note that implementation of the block-encoding for the dipole operator $\hat\mu$ will depend on how the particular system of interest is being represented. For a concrete example we consider a resonant Raman X-ray scattering event, where the photon scattering is caused by electron-photon interactions. In this case, the associated dipole entering the \gls{GQPE} algorithm corresponds to the electronic dipole operator, alongside the electronic structure Hamiltonian. The dipole here is a one-electron operator, which means that it can be efficiently diagonalized using orbital rotations and implemented using a linear combination of unitaries approach with an optimal 1-norm \cite{LCU4}.

\begin{figure}
    \centering
    \[
\Qcircuit @C=2.5em @R=1em {
    &\\
    &\lstick{\ket{0}} & \gate{U_{\mathcal{L}_{\delta}}} & \qw & \qw & \qw & \gate{} \qwx[2] & \gate{\textrm{QFT}^\dagger} & \meter & \rstick{\omega_1}  \\
    &\lstick{\ket{0}} & \gate{U_{\mathcal{L}_{\eta_{int}}}} & \qw & \gate{}\qwx[1] & \qw & \qw & \gate{\textrm{QFT}^\dagger} & \meter  & \rstick{\omega_2}  \\
     &\lstick{\ket{0}} & \gate{U_{\lambda_{n_0}}} & \gate{\mu} & \gate{e^{-i\hat H t_2}} & \gate{\mu} & \gate{e^{-i\hat H t_1}} & \qw & \qw
}
\]
    \caption{Circuit for simulating resonant Raman spectroscopy. Note that the $\mathcal{L}_\delta$ line shape should be chosen so that its square approximates a Dirac delta as quickly as possible, choosing e.g. a Kaiser window or a Gaussian line shape for it (see \cref{app:lineshapes}). Binary encodings of the frequencies $\omega_1$ and $\omega_2$ are sampled with the probability distribution $|R(\omega_1,\omega_2)|^2$ from which the resonant Raman cross section can be obtained, as discussed in \cref{eq:gqpe_to_raman}.}
    \label{fig:raman}
\end{figure}

\subsection{Single-ancilla algorithm}\label{sec:MC}

\begin{figure}[t!]
\centering
\[
\Qcircuit @C=2.5em @R=1em {
    &\\
     &\lstick{\ket{0}} & \gate{H} & \ctrl{1} & \gate{W} & \gate{H} & \meter & \rstick{p(0) - p(1) = \Omega^{-1}\mathbf{Re}/\mathbf{Im}\left[C(\vec t) \right]} \\
     &\lstick{\ket{0}} & \gate{U_{\lambda_{n_0}}} & \gate{C(\vec t)} & \qw & \qw & \qw
}
\]
\caption{Hadamard-based circuit for Monte Carlo single-ancilla approach. The gate $W$ corresponds to $I/S^\dagger$ for obtaining the real/imaginary component. The circuit implementing $\hat C(\vec t)$ is the same as that appearing in \cref{fig:corr_func}, where instead of having a multiplexing over the time registers we have a controlled application over a particular set of times $\vec t$.}
\label{fig:hadamard}
\end{figure}
We now propose a more hardware-friendly implementation of \gls{GQPE} which generalizes the work by Lin and Tong \cite{lin_and_tong} for \gls{QPE}. This replaces the additional time/frequency registers required by \gls{GQPE} with a single ancilla qubit. As shown in \cref{eq:gen_path}, the computed quantity by \gls{GQPE} consists of a convolution between line shape functions and a correlation function, which through the convolution theorem becomes Eq.~\eqref{eq:td_resp}. We start by noting that the frequency-dependent correlation function $C(\vec\omega)$ will be $0$ whenever some $|\omega_j| > \pi$. We can thus consider all associated frequencies to go from $\omega_j\in[-\pi,\pi)$ as opposed to $(-\infty,\infty)$. A key assumption of a localized line shape is made at this point, which can be translated as $\mathcal{L}(\omega) \approx 0 \ \forall |\omega|>\pi$. This allows us to approximate the line shape by a periodic function in the interval $[-\pi,\pi)$. The function $\mathcal{L}(\omega)$ can then be expressed as a Fourier series instead of a Fourier integral:
\begin{equation} \mathcal{L}(\omega) \approx \sum_{k=-\infty}^\infty l(k) e^{-i k\omega}. \end{equation} 
Here, $l(t) = \ift{\mathcal{L}(\omega)}$. The information coming from the line shape is used to avoid having to obtain the correlation function with a high accuracy in exponentially many time points. We now define a probability distribution for $\vec k$ as
\begin{equation} \label{eq:time_dist}
    P(\vec k) = \frac{1}{P_{tot}} \prod_{i=1}^D l(k_i),
\end{equation}
where $P_{tot}$ is a normalizing constant such that $\sum_{k_1=-\infty}^\infty \cdots \sum_{k_D=-\infty}^\infty P(\vec k)=1$. Consider the random variables $x(\vec t)$ and $y(\vec t)$ such that $\esp{x(\vec t)} = \Omega^{-1}\re{C(\vec t)}$ and $\esp{y(\vec t)} = \Omega^{-1}\im{C(\vec t)}$, which can be realized by the Hadamard test circuit shown in Fig.~\ref{fig:hadamard} by assigning a value of $1/-1$ for a measurement of the ancilla qubit in $0/1$ and a successful block-encoding of all $\hat V^{(j)}$'s. Here $\Omega=\prod_{j=0}^D |\hat V^{(j)}|_1$ is the product of 1-norms. The full procedure for estimating $R(\vec\omega)$ then becomes:
\begin{enumerate}
    \item Define a target accuracy $\epsilon$. Deduce the required number of samples $M_{\epsilon}\sim\onot{\Omega^2(\epsilon\eta^D)^{-2}}$. Set $m=1$. 
    \item Sample $\vec k^{(m)}$ from the distribution in Eq.~\eqref{eq:time_dist}.
    \item Run the Hadamard test circuit for the 1-norm scaled correlation function $\Omega^{-1} C(\vec k^{(m)})$ shown in Fig.~\ref{fig:hadamard}. Record the resulting random variable as $x_m$. A measurement of $0/1$ in the ancilla qubit corresponds to a value of $1/-1$. Repeat the procedure for the imaginary part, recording the result in the random variable $y_m$
    \item Repeat steps $2$ and $3$ a total of $M_\epsilon$ times.
    \item Obtain the $\epsilon$-approximation to the target quantity as:
\begin{align}
    R_\epsilon(\vec\omega) &= \frac{\Omega P_{tot}}{M_\epsilon} \sum_{m=1}^{M_\epsilon} (x_m+iy_m) e^{-i \vec\omega \cdot \vec k^{(m)}}. \label{eq:fin_mc}
\end{align}
\end{enumerate}
A more detailed deduction of the computational complexity of this procedure is presented in \cref{app:complexities}.  On average, each run will require $\onot{D/\eta}$ calls to the $e^{\pm i \hat H}$ time-evolution oracle.\\

As mentioned when obtaining the response coming from all possible paths in Liouville space, paths always appear in pairs with their complex conjugates, having an interference for which only the real/imaginary part survives for $D$ odd/even. For a path $\alpha$ and its associated conjugate path $\bar\alpha$, we can then write the sum of contributions to the response function as
\begin{align}
    R^{(\alpha)}(\vec\omega) + R^{(\bar\alpha)}(\vec\omega) &= \mathcal{F}^D_{\vec t \rightarrow \vec\omega} \left\{\jj{C^{(\alpha)}(\vec t)} \prod_{i=1}^D l(t_i) \right\},
\end{align}
where we have defined $\jj{\cdot} = \re{\cdot}/i\im{\cdot}$ for $D$ odd/even. Contributions from these two paths can then be obtained by only considering the real/imaginary components that are obtained by the Hadamard test.

\glsreset{GQPE}
\section{Discussion} \label{sec:discussion}
We have presented the \gls{GQPE} framework, which is an extension of \gls{QPE} for multi-variable phase estimation of interaction picture operators. By viewing \gls{QPE} as the Fourier transform of the expectation $\bra{\psi}e^{i\hat H t}\ket{\psi}$, \gls{GQPE} is a generalization to a multi-variate Fourier transform of multi-variate expectations, such as higher-order correlation functions $\bra{\psi}\hat A_I(t_1)\hat B_I(t_2)\hat C_I(t_3)\hat D \ket{\psi}$. As an application we show how \gls{GQPE} enables the computation of arbitrary order response properties such as those used to model nonlinear spectroscopy experiments. Furthermore, \gls{GQPE}'s flexibility to use different window functions allows for incorporation of different spectral broadenings that appear as a result of dissipative environments.\\

The cost for accurate computation of response properties quickly becomes prohibitive for classical computers, particularly for nonlinear processes. This is due to the fact that calculations of response properties requires access to either excited states or the time-evolution operator both of which become intractable as the system size grows. This elevates the significance of the \gls{GQPE} framework to overcome the limitations of classical techniques in computing response quantities, with circuit costs scaling linearly with respect to the response order. \\

Finally, we presented a more hardware-friendly modification of our framework using the ideas proposed by Lin and Tong \cite{lin_and_tong} for lowering the hardware requirements of \gls{QPE} \cite{gsee}. Thus, we could avoid the need to have a large number of qubits encoding times/frequencies for implementations on early fault-tolerant quantum computers. One open question is how to adapt the \gls{GQPE} framework presented here to work with the qubitized walk operator as opposed to time-evolution operator, analogous to the walk-based implementation of \gls{QPE} \cite{walk_based_qpe, qrom}. Another interesting avenue that is left as future work is a full exploration of more complex dissipative effects by replacing time-evolution with a corresponding Liouvillian evolution.\\

\bibliography{biblio}

\newpage
\appendix
\glsresetall

\section{Spectral broadening and lineshapes} \label{app:lineshapes}
In this appendix we give a more in-depth discussion of different lineshapes and how they are related to different physical processes. Our starting point will be the expression for a path response function
\begin{equation}
    R_{AV}^{(D,\alpha)}(\vec\omega) = \left(\prod_{j=1}^D \Theta(\omega_j) \mathcal{L}(\omega_j)  \right)*^D C^{(\alpha)}(\vec\omega).
\end{equation}
Here we have explicitly separated the lineshape into two components: 
\begin{align}
    \Theta(\omega) &\equiv \ft{\theta(t)} \\
    \mathcal{L}(\omega) &\equiv \ft{e^{-\gamma(t)}}.
\end{align}
The Heaviside $\Theta$ component enforces causality, while the $\mathcal{L}$ part has the information from the dissipative environment and causes an energy broadening.

\subsection{Causality and the Kramers-Kronig relations}
We now discuss the effect from the causality component. Obtaining the Fourier transform of the Heaviside function can be done through the following analytical continuation:
\begin{align}
    \Theta(\omega) &= \lim_{\epsilon\rightarrow 0} \ft{e^{-\epsilon |t|} \theta(t)} \\
    &= \lim_{\epsilon\rightarrow 0} \frac{-i}{\omega - i\epsilon} \\
    &= \frac{1}{2\pi}\left(\delta(\omega) - i\mathcal{P}\frac{1}{\omega}\right),
\end{align}
where $\mathcal{P}$ corresponds to the Cauchy principal value: this function is singular at $\omega=0$, and special care needs to be taken for it to be well defined at this point. The path response functions then consist of a convolution with $\Theta(\omega_k)$ over each frequency $\omega_k$. This effectively enforces the Kramers-Kronig relation \cite{mukamel, boyd, kk1, kk2} in each path response for every $k=1,...,D$:
\begin{align}
    R^{(D,\alpha)}_{AV}(\omega_1,\cdots,\omega_k,\cdots,\omega_D) &= -\frac{i}{\pi}\mathcal{P}\int_{-\infty}^\infty d\tilde\omega_k \frac{R^{(D,\alpha)}_{AV}(\omega_1,\cdots,\tilde\omega_k,\cdots,\omega_D)}{\tilde\omega_k-\omega_k}. \label{eq:kk_1}
\end{align}
Separating the real and imaginary parts of the equation and yields another form for this relation:
\begin{align}    
    \re{R_{AV}^{(D,\alpha)}(\omega_1,\cdots,\omega_k,\cdots,\omega_D)} &= \frac{1}{\pi}\mathcal{P}\int_{-\infty}^\infty d\tilde\omega_k \frac{R^{(D,\alpha)}_{AV}(\omega_1,\cdots,\tilde\omega_k,\cdots,\omega_D)}{\tilde\omega_k-\omega_k}  \label{eq:kk_2} \\
    \im{R_{AV}^{(D,\alpha)}(\omega_1,\cdots,\omega_k,\cdots,\omega_D)} &= -\frac{1}{\pi}\mathcal{P}\int_{-\infty}^\infty d\tilde\omega_k \frac{R^{(D,\alpha)}_{AV}(\omega_1,\cdots,\tilde\omega_k,\cdots,\omega_D)}{\tilde\omega_k-\omega_k}. \label{eq:kk_3} 
\end{align}
Note that these relationships are also directly applicable to the total response $\chi_{AV}^{(D)}(\vec\omega)$, which is how they are usually expressed.

\subsection{Dissipation and spectral broadening}
We now discuss how dissipative effects affect the response function. A formal inclusion of dissipative effects requires an open-system treatment, which entails replacing the unitary time evolution under the Hamiltonian with a non-unitary Liouvillian evolution of a density matrix. Several works have been done to implement these non-unitary evolutions on quantum computers \cite{open1, open2, open3, open4, open5, open6}, which could be used to introduce general environmental effects into our response framework. However, the inclusion of arbitrary environmental effects is beyond the scope of this work. \\

A large class of dissipation effects can be added in an \textit{ad hoc} way by modifying the time evolution operator $e^{\pm i\hat H t}\rightarrow e^{\pm i\hat H t}e^{-\hat\Gamma(t)}$. Here, $\hat\Gamma(t)$ is some operator which commutes with $\hat H$ for every time $t$, which we write as $\hat\Gamma(t) = \sum_n \gamma_n(t) \ket{\lambda_n}\bra{\lambda_n}$. This effectively enforces some dissipation rate $\gamma_n(t)$ at time $t$ for each eigenstate $\ket{\lambda_n}$. An example of where this type of dissipation can appear is for solvated molecules, where the $\gamma_n$'s will be related to the spectral density of the solvent \cite{franco_deco}. However, taking into account the full operator $\hat\Gamma(t)$ would also require non-unitary dynamics. We will then make a uniform broadening approximation: we assume the dissipation is the same for all states. This allows us to replace $\hat\Gamma(t)$ by a scalar $\gamma(t)$. Note that some of the energy-dependency of $\hat\Gamma(t)$ can be reintroduced by using different $\gamma(t)$ functions over different parts of the process. This is particularly useful for resonant processes where intermediary states are typically localized in energy, as discussed for resonant Raman spectroscopy in Eq.\eqref{eq:raman}. \\

Table~\ref{tab:lineshapes} summarizes different lineshapes which are commonly encountered and their physical significance. Finally, most lineshapes discussed here, except for the Heaviside window, converge to a Dirac delta in the frequency-domain when the proper limit for their defining parameter is considered.
\begin{table}[]
    \centering
    \begin{tabular}{|c|c|c|p{4cm}|}
        \hline
        Name of lineshape & Time-domain window & Frequency-domain $\mathcal{L}(\omega)$ & Physical significance \\ \hline \hline
        Causality window & $\theta(t) \equiv \begin{cases}
            0,\ t < 0 \\
            1,\ t \geq 0
        \end{cases}$ & $\Theta(\omega) = \frac{1}{2\pi}\left(\delta(\omega) - i\mathcal{P}\frac{1}{\omega} \right)$ & Heaviside function, enforces causality in response function and causes Kramers-Kronig relations [Eqs.~(\ref{eq:kk_1}-\ref{eq:kk_3})]. \\ \hline
        Rectangular window & $\gamma_T(t) = \begin{cases}
            0,\ |t|\leq T/2 \\
            \infty,\ |t|> T/2
        \end{cases}$ & $\mathcal{L}_T(\omega) = T\textrm{sinc}\left(\frac{\omega T}{2}\right)$ & This window often appears when truncating the infinite-time Fourier integral to some finite time. Ubiquitous in QPE. \\ \hline
        Lorentzian window & $\gamma_\eta(t) = \eta|t|$ & $ \mathcal{L}_\eta(\omega) = \frac{1}{\pi}\frac{1}{i\omega+\eta}$ & Simple dissipation with a linear rate. Often appears in spectroscopy. \\ \hline
        Gaussian window & $\gamma_\sigma(t) = \frac{t^2}{\sigma^2}$ & $\mathcal{L}_\sigma(\omega) = \frac{\sigma}{2}e^{-\frac{\sigma^2\omega^2}{4}}$ & Exponentially decaying tails make it an attractive choice for fast convergence and small evolution time cutoffs. When combined with Lorentzian window it yields the Voigt lineshape, which is often encountered in spectroscopy. \\ \hline
        Kaiser window & $e^{-\gamma_{\alpha,L}(t)} = \begin{cases}
        \frac{1}{L} \frac{I_0\left(\pi\alpha\sqrt{1-\left(\frac{2t}{L}\right)}\right)}{I_0(\pi\alpha)},\textrm{ if }|t|\leq \frac{L}{2}, \\
        0,\textrm{ otherwise}.
    \end{cases} $ & $\mathcal{L}_{\alpha,L}(\omega) = \frac{\textrm{sinc}\left(\pi\sqrt{L^2\omega^2 - \alpha^2}\right)}{I_0(\pi\alpha)}$ & Close to optimal frequency concentration around the main $\omega=0$ peak, this kernel achieves the fastest convergence for QPE \cite{vs_qsvt}. $I_0$ corresponds to the zeroth-order modified Bessel function of the first kind. \\ \hline
    \end{tabular}
    \caption{Different lineshapes commonly encountered when calculating response properties.}
    \label{tab:lineshapes}
\end{table}

\section{Proof of action of circuits} \label{app:proof}
\begin{figure}[t!]
\centering
\[
\Qcircuit @C=1.5em @R=1em {
     & \gate{H}  & \ctrlo{1} &  \ctrl{1} & \gate{H} & \qw \\
     & \qw &\multigate{2}{\hat U_{R(\vec\omega)}} & \multigate{2}{\hat U_{R(\vec\omega)}^*} & \qw & \qw\\
     & \vdots & & & & &\\
     & \qw & \ghost{\hat U_{R(\vec\omega)}} & \ghost{\hat U_{R(\vec\omega)}^*}& \qw& \qw
}
\]
\caption{Quantum circuit to separate the real $\frac{\hat U_{R(\vec\omega)} + \hat U_{R(\vec\omega)}^*}{2}$ and the imaginary $\frac{\hat U_{R(\vec\omega)} - \hat U_{R(\vec\omega)}^*}{2}$ parts of $R(\vec{\omega})$ so they can be estimated.}
\label{fig:real_complex_circ}
\end{figure}
In this appendix we show a proof of how the expectation values of the operators introduced in \cref{sec:GQPE} indeed yield discretized approximations of the sought after response functions. This corresponds to the equalities:
\begin{align}
    R(\omega_1,...,\omega_D) &= \bra{\omega_1,...,\omega_D, 0}\hat U_{R(\vec\omega)}\ket{0,\cdots,0, 0},\\
    R(\omega_1,...,\omega_D)^* &= \bra{\omega_1,...,\omega_D, 0}\hat U_{R(\vec\omega)}^*\ket{0,\cdots,0, 0}.
\end{align}
Given
\begin{align}
    \hat U_{R(\vec\omega)} &= \left[\bigotimes_{j=1}^D\text{QFT}^\dag \otimes \hat U_{\lambda_{n_0}}^\dag \right] \left[\sum_{t_1,\cdots t_D} \bigotimes_{j=1}^D \ket{t_j}\bra{t_j} \otimes \hat V^{(0)}_I(t_D) \hat V^{(1)}_I(t_{D-1}) \cdots \hat V^{(D-1)}_I(t_1) \hat V^{(D)}\,\,\right]\left[\bigotimes_{j=1}^D\hat U_{\mathcal{L}} \otimes \hat U_{\lambda_{n_0}}\right],\\
    \hat U_{R(\vec\omega)}^* &= \left[\bigotimes_{j=1}^D\text{QFT}\,\, \otimes \hat U_{\lambda_{n_0}}^\dag\right] \left[\sum_{t_1,\cdots t_D} \bigotimes_{j=1}^D \ket{t_j}\bra{t_j} \otimes \left(\hat V^{(0)}_I(t_D) \hat V^{(1)}_I(t_{D-1}) \cdots \hat V^{(D-1)}_I(t_1) \hat V^{(D)}\right)^\dag\right]\left[\bigotimes_{j=1}^D\hat U_{\mathcal{L}^*}  \otimes \hat U_{\lambda_{n_0}} \right].
\end{align}
Notice $\hat U_{R(\vec\omega)}^* \neq \hat U_{R(\vec\omega)}^\dagger$. Here we have used $\hat U_{\mathcal{L}_{\eta}} \ket{0} = \sum_{k=0}^{N-1} \alpha_k \ket{k}$, and $\hat U_{\mathcal{L}_{\eta}^*} \ket{0} = \sum_{k=0}^{N-1} \alpha_k^* \ket{k}$ as the unitaries that prepare the initial superposition of the time register and $\hat U_{\lambda_n}\ket{0} = \ket{\lambda_n}$ as the unitary that prepares the initial state of the system. We refer the reader to~\cite{ini_state,dynamical_cooling} for a more in-depth discussion on initial state preparation.\\
We write the unitary
\begin{align}
    \hat U_{R(\vec\omega)}\ket{0, 0} &= \left[\bigotimes_{j=1}^D\text{QFT}^\dag \otimes \hat U_{\lambda_{n_0}}^\dag \right] \left[\sum_{t1,\cdots, t_D} \bigotimes_{k=1}^D \ket{t_k}\bra{t_k} \otimes \hat V^{(0)}_I(t_D) \hat V^{(1)}_I(t_{D-1}) \cdots \hat V^{(D-1)}_I(t_1) \hat V^{(D)}\,\,\right]\bigotimes_{j=1}^D\sum_{t_j=0}^{N-1}\alpha_{t_j}\ket{t_j} \otimes \ket{\lambda_{n_0}}\nonumber\\
    &= \left[\bigotimes_{j=1}^D\text{QFT}^\dag \otimes \hat U_{\lambda_{n_0}}^\dag \right] \sum_{t_1,\cdots t_D} \bigotimes_{k=1}^d\alpha_{t_k}\ket{t_k} \otimes \hat V^{(0)}_I(t_D) \hat V^{(1)}_I(t_{D-1}) \cdots \hat V^{(D-1)}_I(t_1) \hat V^{(D)}\ket{\lambda_{n_0}}.
\end{align}
Therefore
\begin{align}
    \bra{\Vec \omega, 0}\hat U_{R(\vec\omega)}\ket{0, 0} &= \bra{\Vec \omega}\left[\bigotimes_{j=1}^D\text{QFT}^\dag\right] \sum_{t_1,\cdots, t_D} \bigotimes_{k=1}^D\alpha_{t_k}\ket{t_k} \cdot \bra{\lambda_{n_0}}\hat V^{(0)}_I(t_D) \hat V^{(1)}_I(t_{D-1}) \cdots \hat V^{(D-1)}_I(t_1) \hat V^{(D)}\ket{\lambda_{n_0}}\nonumber\\
    &= \bra{\Vec \omega}\left[\bigotimes_{j=1}^D\text{QFT}^\dag\right] \sum_{t_1,\cdots, t_D} \bigotimes_{k=1}^D\alpha_{t_k}\ket{t_k} \sum_{n_1,...,n_D} V^{(D)}_{n_Dn_0}\prod_{j=1}^{D} V^{(j-1)}_{n_{j-1}n_{j}}\,e^{i\Delta_{(j-1)j}\cdot t_{D-j+1}}\nonumber\\
    &= \sum_{t_1,\cdots, t_D} \prod_{k=1}^D \frac{\alpha_{t_k}}{\sqrt{N}}\,e^{-it_k\cdot \omega_k \cdot 2\pi/N} \sum_{n_1,...,n_D} V^{(D)}_{n_Dn_0}\prod_{j=1}^{D} V^{(j-1)}_{n_{j-1}n_{j}}\,e^{i\Delta_{(j-1)j}\cdot t_{D-j+1}}\nonumber\\
    &= \sum_{n_1,...,n_D} V^{(D)}_{n_Dn_0}\prod_{j=1}^{D} V^{(j-1)}_{n_{j-1}n_{j}} \Biggl[\sum_{t_1=0}^{N-1} \frac{\alpha_{t_1}}{\sqrt{N}}\,e^{i(\Delta_{(D-1)D} - \frac{\omega_1 \cdot 2\pi}{N})t_1} \sum_{t_2=0}^{N-1} \frac{\alpha_{t_2}}{\sqrt{N}}\,e^{i(\Delta_{(D-2)(D-1)} - \frac{\omega_2 \cdot 2\pi}{N})t_2}\nonumber\\
    &\quad\cdots\sum_{t_{D-1}=0}^{N-1} \frac{\alpha_{t_{D-1}}}{\sqrt{N}}\,e^{i(\Delta_{12} - \frac{\omega_{D-1} \cdot 2\pi}{N})t_{D-1}}\sum_{t_D=0}^{N-1} \frac{\alpha_{t_D}}{\sqrt{N}}\,e^{i(\Delta_{01} - \frac{\omega_D \cdot 2\pi}{N})t_D}\Biggr]
\end{align}
Since $\mathcal{L}_{\eta}(\omega) = \frac{1}{\sqrt{N}}\sum_{k=0}^{N-1}\alpha_k e^{ik\omega}$, then
\begin{align}
    \bra{\Vec\omega, 0}\hat U_{R(\vec\omega)}\ket{0, 0} &= \sum_{n_1,...,n_D} V^{(D)}_{n_Dn_0}\prod_{j=1}^{D} V^{(j-1)}_{n_{j-1}n_{j}}\, \mathcal{L}_{\eta}\left(\Delta_{(j-1)j}-\frac{2\pi\cdot\omega_{D-j+1}}{N}\right).
\end{align}
The equality for $\hat U_{R(\vec\omega)}^*$ can be shown via a similar analysis. From these expectation values, it becomes clear how using the standard linear combination of unitaries circuit in Fig.~\ref{fig:real_complex_circ} we can recover the real and imaginary components of the response function. Note that this approach yields the absolute values for these quantities. One possible approach for recovering the signs would be to implement the same circuits as in \cref{fig:real_complex_circ} with a global shift as shown in Ref.~\citenum{real_qae}.

\section{Complexity of single-ancilla algorithm} \label{app:complexities}
Here, we discuss the cost for the Monte Carlo single-ancilla method. We start by considering the $D=1$ one-dimensional case, after which we extend our deduction to the general multi-dimensional case. \\

We first calculate the expectation value of $k$ which is chosen according the probability $P(k)/P_{tot}$ in \cref{eq:time_dist}. For simplicity we consider a Lorentzian broadening with width $\eta$, noting that other broadenings will give a very similar width-dependent complexity \cite{gsee}. We thus have $P(k) = \theta(k) e^{-\eta k}$:
\begin{equation}
    \mathbb{E}(k) = \frac{1}{P_{tot}^{(1)}}\sum_{k=-\infty}^\infty P(k) k = \frac{1}{P_{tot}^{(1)}}\sum_{k=-\infty}^\infty \theta(k)e^{-\eta k} k.
\end{equation}
The scalings for the two elements of this expectation are:
\begin{equation}
\begin{aligned}
    P_{tot}^{(1)} = \sum_{k=0}^\infty e^{-\eta k} &= \frac{1}{1- e^{- \eta}} \sim \mathcal{O}\left(\frac{1}{\eta}\right), \\
    \sum_{k=0}^\infty e^{-\eta k} k &=  -\frac{d}{d\eta} \left( \sum_{k=0}^\infty e^{-\eta k}  \right) = \frac{e^{-\eta}}{(1-e^{-\eta})^2} \sim \mathcal{O}\left(\frac{1}{\eta^2}\right).
\end{aligned}    
\end{equation}
This results in:
\begin{equation}
    \mathbb{E}(k) \sim \mathcal{O}\left( \frac{1}{\eta} \right).
\end{equation}
On the other hand, if we take the outcomes of each single measurement, since they will be $\pm 1$ values, their variance can be calculated as follows:
\begin{equation}
    \text{var} ( x_m+iy_m ) \leq \mathbb{E} (|x_m|^2+|y_m|^2) = 2
\end{equation}
The variance in the whole random variable which results in the response function shown in \cref{eq:fin_mc} is given by $(\Omega P^{(1)}_{tot})^2 \text{var}( x_m+iy_m )$. The total number of repetitions in order to obtain an error $\epsilon$ thus becomes
\begin{equation}
    M_\epsilon^{(1)} \sim \mathcal{O} \left( \frac{(\Omega P_{tot}^{(1)})^2}{\epsilon^2} \right) =
     \mathcal{O} \left( \frac{\Omega^2}{\eta^2 \cdot \epsilon^2 } \right),
\end{equation}
where $\Omega\equiv \prod_{j=0}^D |\hat V^{(j)}|_1$ encodes the 1-norms coming from the block-encodings of all $\hat V^{(j)}$'s. In order to find the total number of queries to the time evolution operator $e^{-iH}$, we multiply the above number by $\mathbb{E}(k)$, which results in:
\begin{equation}
    \mathcal{O}\left( \frac{\Omega^2}{ \eta^3 \cdot\epsilon^2 } \right)
\end{equation}

We are now ready to deduce the general cost for the multi-dimensional case, having an arbitrary order $D$. We first note that the average number of queries in each round follows:
\begin{equation}
    \mathbb{E}\left[ \sum_{i=1}^D k_i \right] = \sum_{i=1}^D \mathbb{E}\left[ k_i \right] \sim \mathcal{O}\left( \frac{D}{\eta} \right).
\end{equation}
Furthermore, the total number of times that the measurement needs to be repeated can be calculated as follows:
\begin{equation}
    M_\epsilon^{(D)} \sim \mathcal{O}\left(\frac{\Omega^2}{\eta^{2D}\cdot \epsilon^2} \right)
\end{equation}
where we have used that the variance of the random variable scales as $(\Omega P_{tot}^{(D)})^2 = \Omega^2 (P_{tot}^{(1)})^{2D}$. The total number of queries to $e^{-iH}$, which is obtained by multiplying by $\mathbb{E}\left[ \sum_{i=1}^D k_i \right]$, then becomes
\begin{equation}
    \mathcal{O}\left( \frac{D\cdot\Omega^2}{\eta^{2D+1} \cdot \epsilon^2} \right).
\end{equation}
Finally, note that all of our previous analysis works with the scaled Hamiltonian such that $||\hat H||\leq \pi$. For a general Hamiltonian, we can define $\tau \sim \onot{||\hat H||^{-1}}$ such that the Hamiltonian that enters the algorithm is $\hat H \tau$. Noting that the response has units $||\hat H||^{-1}$, obtaining an accuracy $\epsilon$ for the scaled Hamiltonian requires an accuracy $\epsilon/\tau$ for the general case. The scaled spectral width $\eta$ becomes $\eta\tau$. The total number of calls to $e^{-i\hat H\tau}$ for estimating the response function of a general Hamiltonian using the single-ancilla approach thus becomes
\begin{equation}
    \onot{\frac{D\cdot\Omega^2 }{\eta^{2D+1}\cdot\tau^{2D-1}\cdot\epsilon^2}}.
\end{equation}
Note that the same $\epsilon\rightarrow \epsilon/\tau$ and $\eta\rightarrow\eta\tau$ scalings should be considered when dealing with a general non-scaled Hamiltonian for the complexities shown for the \gls{GQPE}-based algorithms in \cref{sec:GQPE}.

\end{document}